\documentclass[11pt]{article}
\usepackage[left=1in,top=1in,right=1in,bottom=1in]{geometry}
\usepackage{times}

\usepackage{amsmath}
\usepackage{amsthm}

\usepackage{siunitx}

\newtheorem{definition}{Definition}
\newtheorem{theorem}{Theorem}

\usepackage{verbatim}

\usepackage{rotating}
\usepackage{algorithm}
\usepackage[noend]{algpseudocode}

\usepackage{tabularx}
\usepackage{subfigure}
\usepackage{graphicx}
\usepackage{url}
\usepackage{bigfoot}
\usepackage{multirow}

\usepackage{silence}
\begin{document}

\title{Optimization-Based Algorithm for Evolutionarily Stable Strategies against Pure Mutations}

\author{Sam Ganzfried\\
Ganzfried Research\\
sam@ganzfriedresearch.com \\
}





\date{\vspace{-5ex}}

\maketitle

\begin{abstract}
Evolutionarily stable strategy (ESS) is an important solution concept in game theory which has been applied frequently to biological models. Informally an ESS is a strategy that if followed by the population cannot be taken over by a mutation strategy that is initially rare. Finding such a strategy has been shown to be difficult from a theoretical complexity perspective. We present an algorithm for the case where mutations are restricted to pure strategies, and present experiments on several game classes including random and a recently-proposed cancer model. Our algorithm is based on a mixed-integer non-convex feasibility program formulation, which constitutes the first general optimization formulation for this problem. It turns out that the vast majority of the games included in the experiments contain ESS with small support, and our algorithm is outperformed by an existing support-enumeration based approach. However we suspect our algorithm may be useful in the future as games are studied that have ESS with potentially larger and unknown support size.
\end{abstract}

\section{Introduction}
\label{se:intro}
While Nash equilibrium has emerged as the standard solution concept in game theory, it is often criticized as being too weak: often games contain multiple Nash equilibria (sometimes even infinitely many), and we want to select one that satisfies other natural properties. For example, one popular concept that refines Nash equilibrium is evolutionarily stable strategy (ESS). A mixed strategy in a two-player symmetric game is an evolutionarily stable strategy if, informally, it is robust to being overtaken by a mutation strategy. Formally, mixed strategy $x^*$ is an ESS if for every mixed strategy $x$ that differs from $x^*$, there exists $\epsilon_0 = \epsilon_0(x) > 0$ such that, for all $\epsilon \in (0,\epsilon_0)$, 
\begin{equation}
(1-\epsilon)u_1(x,x^*) + \epsilon u_1(x,x) <  (1-\epsilon)u_1(x^*,x^*)+\epsilon u_1(x^*,x).
\label{eq:ess}
\end{equation}
From a biological perspective, we can interpret $x^*$ as a distribution among ``normal'' individuals within a population, and consider a mutation that makes use of strategy $x$, assuming that the proportion of the mutation in the population is $\epsilon$. In an ESS, the expected payoff of the mutation is smaller than the expected payoff of a normal individual, and hence the proportion of mutations will decrease and eventually disappear over time, with the composition of the population returning to being mostly $x^*$. An ESS is therefore a mixed strategy of the column player that is immune to being overtaken by mutations. ESS was initially proposed by mathematical biologists motivated by applications such as population dynamics (e.g., maintaining robustness to mutations within a population of humans or animals)~\cite{Smith73:Logic,Smith82:Evolution}. A common example game is the 2x2 game where strategies correspond to an ``aggressive'' Hawk or a ``peaceful'' Dove strategy. A paper has recently proposed a similar game in which an aggressive malignant cell competes with a passive normal cell for biological energy, which has applications to cancer eradication~\cite{Dingli09:Cancer}. 

While Nash equilibrium is defined for general multiplayer games, ESS is defined specifically for two-player symmetric games. ESS is a refinement of Nash equilibrium. In particular, if $x^*$ is an ESS, then $(x^*,x^*)$ (i.e., the strategy profile where both players play $x^*$) is a (symmetric) Nash equilibrium~\cite{Maschler13:Game}. Of course the converse is not necessarily true (not every symmetric Nash equilibrium is an ESS), or else ESS would be a trivial refinement. In fact, ESS is not guaranteed to exist in games with more than two pure strategies per player (while Nash equilibrium is guaranteed to exist in all finite games). For example, while rock-paper-scissors has a mixed strategy Nash equilibrium (which puts equal weight on all three actions), it has no ESS~\cite{Maschler13:Game} (that work considers a version where payoffs are 1 for a victory, 0 for loss, and $\frac{2}{3}$ for a tie).

There exists a polynomial-time algorithm for computing Nash equilibrium in two-player zero-sum games, while for two-player general-sum and multiplayer games computing a Nash equilibrium is PPAD-complete~\cite{Chen05:Nash,Chen06:Settling} and it is widely conjectured that no efficient (polynomial-time) algorithm exists. However, several algorithms have been devised that perform well in practice~\cite{Berg17:Exclusion,Porter08:Simple,Govindan03:Global,Sandholm05:Mixed,Lemke64:Equilibrium}. For ESS, there have been some recent hardness results as well (note that in general computing an ESS is at least as hard as computing a Nash equilibrium since it is a refinement). The problem of computing whether a game has an ESS was shown to be both NP-hard and CO-NP hard and also to be contained in $\Sigma ^P _2$ (the class of decision problems that can be solved in nondeterministic polynomial time given access to an NP oracle)~\cite{Etessami08:Computational}. Subsequently it was shown that the exact complexity of this problem is that it is $\Sigma ^P _2$-complete~\cite{Conitzer13:Exact}. Note that this result is for the complexity of determining whether an ESS exists in a given game (as discussed above there exist games which have no ESS), and not for the complexity of actually computing an ESS in games for which one exists. 

Several approaches have been proposed for computing ESS based on performing a brute-force enumeration of supports, solving for a potential ESS with each support in turn (in theory these algorithms find all ESS when the enumeration is continued over all possible supports)~\cite{Haigh75:Game,Abakuks80:Conditions,Broom13:Game-Theoretical,Bomze92:Detecting,Mcnamara97:General}. While these approaches would perform very well for games that have an ESS with small support, they can potentially run very slowly for games that have ESS with large or unknown support size, as there are $\binom{n}{k}$ total supports of size $k$ in a game with $n$ pure strategies which is exponential. (Note also that these algorithms could alternatively enumerate the supports from largest to smallest for games known to have ESS with large support. The main challenging case would be if there exist ESS with unknown support of medium size. Furthermore, for games with no ESS, which are possible for more than two pure strategies, these approaches will take a long time to determine that no solution exists.) Similarly for Nash equilibrium computation, the support enumeration algorithm~\cite{Porter08:Simple} has been demonstrated to outperform approaches based on optimization frameworks~\cite{Sandholm05:Mixed,Lemke64:Equilibrium} on many standard test games, which contain equilibria with small support; however the optimization-based approaches perform better on games that have equilibria with medium-sized support. 

We present an algorithm that computes an ESS that satisfies a less restrictive condition that only pure strategy mutations cannot successfully invade the population (while standard ESS rules out all mixed strategy mutations). The algorithm is based on modeling the problem as a mixed-integer non-convex quad\-rat\-ically-constrained feasibility program. We present experiments on uniformly-random games, games generated from the GAMUT repository, and a recently-proposed cancer model. As it turns out, for all of these classes there typically exists an ESS with small support, and therefore it is not surprising that the support enumeration algorithm outperforms our optimization-based formulation for these games. For games that ESS with unknown medium support size, we expect our new algorithm would outperform the enumeration ones. Note also that our algorithm has several parameters that can be modified to enable the algorithm to tradeoff between speed and performance, while the enumeration algorithms have no parameters and will take too long once the number of supports ($2^n$) gets large. In the future we plan to track down several classes of important games that have ESS with medium-sized supports to better ascertain the relative performance between the algorithms.

\section{Evolutionarily stable strategies against pure mutations}
\label{se:ess}
We define an \emph{Evolutionarily Stable Strategy against Pure Mutations} (ESSPM) as in Equation~\ref{eq:ess} except that we only require that the inequality holds for all pure strategies $x$ that differ from $x^*$. In order to model the problem of computing an ESSPM as an optimization problem, we recall a widely-known alternative definition of ESS that has been proven to be equivalent to the initial one~\cite{Maschler13:Game}.


\begin{definition}
\label{de:ess}
A mixed strategy $x^*$ in a two-player symmetric game is an evolutionarily stable strategy (ESS) if for every mixed strategy $x$ that differs from $x^*$ there exists $\epsilon_0 = \epsilon_0(x) > 0$ such that, for all $\epsilon \in (0,\epsilon_0)$,
$$(1-\epsilon)u_1(x,x^*) + \epsilon u_1(x,x) < (1-\epsilon)u_1(x^*,x^*)+\epsilon u_1(x^*,x).$$
\end{definition}

\begin{theorem}
A strategy $x^*$ is evolutionarily stable if and only if for each $x \neq x^*$ exactly one of the following conditions holds:
\begin{itemize}
\item $u_1(x,x^*) < u_1(x^*,x^*)$ 
\item $u_1(x,x^*) = u_1(x^*,x^*)$ and $u_1(x,x) < u_1(x^*,x)$
\end{itemize}
\label{th:ess}
\end{theorem}

\begin{definition}
\label{de:esspm}
A mixed strategy $x^*$ in a two-player symmetric game is an evolutionarily stable strategy against pure mutations (ESSPM) if for every pure strategy $x$ that differs from $x^*$ there exists $\epsilon_0 = \epsilon_0(x) > 0$ such that, for all $\epsilon \in (0,\epsilon_0)$,
\begin{equation}
(1-\epsilon)u_1(x,x^*) + \epsilon u_1(x,x) < (1-\epsilon)u_1(x^*,x^*)+\epsilon u_1(x^*,x).
\label{eq:esspm}
\end{equation}
\end{definition}

\begin{theorem}
A strategy $x^*$ is evolutionarily stable against pure mutations if and only if for each pure strategy $x \neq x^*$ exactly one of the following conditions holds:
\begin{itemize}
\item $u_1(x,x^*) < u_1(x^*,x^*)$ 
\item $u_1(x,x^*) = u_1(x^*,x^*)$ and $u_1(x,x) < u_1(x^*,x)$
\end{itemize}
\label{th:esspm}
\end{theorem}

The proof of Theorem~\ref{th:esspm} follows from similar reasoning as the proof of Theorem~\ref{th:ess}~\cite{Maschler13:Game}.

(Note that technically the statements of Theorem~\ref{th:ess} and~\ref{th:esspm} can be logically weakened to require that only ``at least'' as opposed to ``exactly'' one of the conditions holds, since it is impossible for both to hold simultaneously.) 

There are also known results that allow us to categorize ESS with respect to the more common solution concept Nash equilibrium~\cite{Maschler13:Game}:

\begin{theorem}
If $x^*$ is a an evolutionarily stable strategy in a two-player symmetric game, then $(x^*,x^*)$ is a symmetric Nash equilibrium of the game.
\label{th:ESS-symm}
\end{theorem}

\begin{theorem}
In a symmetric game, if $(x^*,x^*)$ is a strict symmetric Nash equilibrium then $x^*$ is an evolutionarily stable strategy.
\label{th:ESS-strict}
\end{theorem}

Recall that in a Nash equilibrium $x^*$, we have that $u_1(x,x^*) \leq u_1(x^*,x^*)$ for all mixed strategies $x$ (and similarly for the other players). In a strict equilibrium this requirement changes to $u_1(x,x^*) < u_1(x^*,x^*).$

Clearly every ESS is also an ESSPM, and so therefore ESS is a refinement of ESSPM. We can straightforwardly show a similar result to Theorem~\ref{th:ESS-symm} for ESSPM. This shows that ESSPM is a refinement of Nash equilibrium.
\begin{theorem}
If $x^*$ is a an ESSPM in a two-player symmetric game, then $(x^*,x^*)$ is a symmetric Nash equilibrium. 
\label{th:ESSPM-symm}
\end{theorem}
\begin{proof}
Consider taking the limit as $\epsilon \rightarrow 0$ in Equation~\ref{eq:esspm}. From the continuity of the utility function (which is a standard assumption in game theory), this gives $u_1(x,x^*) \leq u_1(x^*,x^*)$ for all pure strategies $x$. This condition is sufficient to show that $x^*$ is a Nash equilibrium; if a player could profitably deviate from $x^*$ to a mixed strategy $x'$, then for all the pure strategies in the support of $x'$ he would obtain higher payoff against $x^*$ than by following $x^*$, which contradicts the fact that $u_1(x,x^*) \leq u_1(x^*,x^*)$ for all pure strategies. So $(x^*,x^*)$ is a Nash equilibrium, and is symmetric because the strategies for the players are the same.  
\end{proof}

We can also show a result similar to Theorem~\ref{th:ESS-strict}:

\begin{theorem}
In a symmetric game, if $(x^*,x^*)$ is a strict symmetric Nash equilibrium then $x^*$ is an ESSPM.
\label{th:ESSPM-strict}
\end{theorem}
\begin{proof}
If $(x^*,x^*)$ is a strict symmetric Nash equilibrium, then $x^*$ satisfies the first condition from Theorem~\ref{th:esspm} for all mixed strategies, and therefore also satisfies the condition for all pure strategies $x$. Therefore, $x^*$ is an ESSPM.
\end{proof}

Table~\ref{fi:counterexample} gives an example of a game that has an ESSPM which is not an ESS, demonstrating that the two are not equivalent and that ESSPM is in fact a strictly weaker solution concept (note that it can be shown relatively easily that ESSPM and ESS are equivalent for 2x2 games). Strategy A is evolutionarily stable against both pure strategies B or C (both attain 2 against A but obtain 0 against themselves while A obtains 1). However, consider the mixed-strategy mutation that plays B and C with probability $\frac{1}{2}$. This strategy obtains expected payoff $\frac{1}{2}2 + \frac{1}{2}2 = 2$ against A (which is the same payoff that A obtains against itself). Against itself this strategy obtains expected payoff $\frac{1}{4}(0 + 4 + 4 + 0) = 2$, while A obtains payoff of $\frac{1}{2}(1 + 1) = 1$ against it. 

\begin{table}[!ht]
\begin{center}
\begin{tabular}{c|*{3}{c|}}
&A &B &C \\ \hline
A &(2,2) &(1,2) &(1,2) \\ \hline
B &(2,1) &(0,0) &(4,4) \\ \hline
C &(2,1) &(4,4) &(0,0) \\ \hline
\end{tabular}
\caption{Strategy A is evolutionarily stable against B or against C alone, but against a mutant $\frac{1}{2}$B + $\frac{1}{2}$C, A does worse than the mutant strategy.}
\label{fi:counterexample}
\end{center}
\end{table}

\section{Algorithm}
\label{se:algo}
Given the formulation from Theorem~\ref{th:esspm}, we can cast the problem of computing an ESSPM as the following feasibility program:
\begin{align}
\label{eq:ESSPM}
&u_1(x,x^*) \leq u_1(x^*,x^*) - \epsilon_{x1} + M_{x1}y_x \mbox{ for all } x \in [1,\ldots,m]\\
&u_1(x,x^*) \leq u_1(x^*,x^*) + M_{x2}(1-y_x) \mbox{ for all } x \in [1,\ldots,m]\\
&u_1(x^*,x^*) \leq u_1(x,x^*) + M_{x3}(1-y_x) \mbox{ for all } x \in [1,\ldots,m]\\
&u_1(x,x) \leq u_1(x^*,x) - \epsilon_{x2} + M_{x4}(1-y_x) \mbox{ for all } x \in [1,\ldots,m]\\
&x^*_i \geq 0 \mbox{ for all } i \in [1,\ldots,m]\\
&\sum_i x^* _i = 1\\
&y_i \mbox{ binary for all } i \in [1,\ldots,m]
\end{align}
\normalsize

We assume that we are initially given an $m \times m$ matrix $A$ of utilities to player 1 (we assume that the game is symmetric since ESS is defined only for symmetric games, so payoffs to player 2 are given by matrix $B = A^T$). If $v$ and $w$ are pure strategies, the expression $u_1(v,w)$ corresponds to the $v$'th row and $w$'th column of matrix $A$. In general if $v$ and $w$ are vectors of mixed strategies (i.e., probability distributions over pure strategies), then the utility for player 1 is given by $u_1(v,w) = v^T A w$. Without loss of generality we can assume that all values in $A$ are nonnegative and between 0 and 1 (note that any affine transformation does not affect strategic aspects of the game, so we can add a sufficiently large constant to all entries and normalize by dividing by the largest entry to achieve this condition).

The variables in the formulation are $x^*_i$ and $y_i$ for $1 \leq i \leq m$. The $x^*_i$ correspond to the ESSPM for player 1 (and equivalently player 2) that we are seeking to compute, and $y_i$ correspond to indicator variables denoting which of the conditions holds from Theorem~\ref{th:esspm} for the given component.

We set $\epsilon$ to be a very small floating point number slightly larger than 0, such as $\epsilon = 0.00001.$ $M$ denotes a constant that exceeds the maximum difference in absolute value of utility between two strategy profiles. We would prefer to have $M$ be as small as possible to make the inequality tighter, so we will set $M = 1 + \epsilon$, given our assumption that all payoffs are between 0 and 1. Note that our framework is flexible enough to allow for different selections for the different $\epsilon_{xi}$ and $M_{xi}$ parameters corresponding to the different constraints if desired, though for simplicity we set them all to be the same $\epsilon$ and $M$ as described. The equivalence of the feasibility program to the conditions of Theorem~\ref{th:esspm} follows from a rule for representing either-or constraints in a mixed-integer program by adding in auxiliary binary indicator variables~\cite{Bisschop06:AIMMS}. If $y_x = 0$, then the first constraint ensures that $u_1(x,x^*) \leq u_1(x^*,x^*) - \epsilon_{x1}$, which is equivalent to $u_1(x,x^*) < u_1(x^*,x^*)$ since $\epsilon_{x1}$ is negligible (note that strict inequalities must be converted to weak inequalities to be solved by most standard optimization algorithms). (Numerical precision considerations will be elaborated on further below.) And if $y_x = 1$ then the constraint states $u_1(x,x^*) \leq u_1(x^*,x^*) - \epsilon_{x1} + M_{x1}$, which essentially makes the constraint inactive since it will be true for all strategies $x^*$ since we have chosen $M_{x1}$ to be sufficiently large. The translation between the other constraints and the theorem can be obtained similarly. Note that this formulation is a quadratically-constrained mixed-integer feasibility program (QCMIP). It is quadratically constrained because $u_1(x^*,x^*)$ equals $x^{*T} A x^*$, which is a quadratic form involving a product of the variables $x^*$. It is mixed integral because there are both continuous and binary variables. And it is a feasibility program because there is no objective function, though there are several candidates that could potentially be used that may help aid performance in practice. 

To implement our algorithm we explored the QVXPY and Gurobi software packages~\cite{Park17:General,Gurobi14:Gurobi} and will present analysis in Section~\ref{se:example}. Beyond finding a suitable optimization solver, there are several further challenges for applying our algorithm. A game may contain no ESSPM, in which case our algorithm should return that the optimization model is infeasible. Due to degrees of approximation error and numerical precision there is some chance of a false negative (our algorithm outputs that the model is infeasible while an ESSPM actually exists). Note that if we were solving for a different solution concept that is known to exist, such as Nash equilibrium, we would know for sure that an output of infeasibility would indicate a false negative; however for ESSPM we cannot be sure whether the game actually has no ESSPM or we are in a false negative (and so our algorithm is \emph{incomplete}). We will explore the extent of this issue in the experiments. A second issue is that according to the alternative definition of ESSPM given by Theorem~\ref{th:esspm}, the conditions apply just to $x \neq x^*$, which we do not encode directly in our algorithm, which considers deviations for all $x$ since it cannot readily tell whether $x = x^*$. This can be problematic for the case when $x^*$ is a pure strategy, and the conditions for $x = x^*$ would be enforced when they should not be (which would lead to a false negative). To address this we will first apply an efficient preprocessing procedure to determine if a pure strategy ESSPM exists, which is described in Algorithm~\ref{al:esspm-pure}. 

\begin{algorithm}[!ht]
\caption{Preprocessing procedure for computing pure-strategy ESSPM}
\label{al:esspm-pure} 
\textbf{Inputs}: Payoff matrix $M$ with $m$ pure strategies per player.
\begin{algorithmic}
\For {$i = 1$ to $m$} 
\For {$j = 1$ to $m$, $j \neq i$} 
\State Test whether the conditions of Theorem~\ref{th:esspm} hold using $x^* = i$, $x = j.$
\EndFor 
\State If the conditions held for all $j$, output $i$ as pure-strategy ESSPM. 
\EndFor 
\State Conditions of Theorem~\ref{th:esspm} did not hold for any pure strategy, so there is no pure ESSPM 
\end{algorithmic}
\end{algorithm}

Since our algorithm has elements of approximation error and numerical imprecision, we must develop a metric to evaluate the quality of the candidate solution. For approximating Nash equilibrium the standard metric is the maximum amount a player can gain by deviating from the candidate equilibrium strategy profile $x^*$, i.e., $\epsilon = \max_i \max_j (u_i(x_j,x^*_{-i}) - u_i(x^*, x^*_{-i}))$ (where $u_i(x,y)$ is the utility to player $i$ when player $i$ follows $x$ and the opponent(s) follow $y$, and $x^*_{-i}$ is the strategy vector for $x^*$ excluding player $i$). It is common for Nash equilibrium algorithms to be analyzed by their convergence with respect to this $\epsilon$~\cite{Gilpin12:First}. As one notable example an $\epsilon$-Nash equilibrium has been computed for two-player limit Texas hold 'em for a sufficiently small $\epsilon$ that a human cannot differentiate between the approximation and an exact equilibrium over a lifetime of human play~\cite{Bowling15:Heads-up}. For ESSPM we propose a similar metric, given by Algorithm~\ref{al:esspm-error}. One challenge is that it may not be clear which of the conditions to consider, since small amounts of numerical instability can make it impossible to satisfy a hard equality. To address this we include a precision parameter $\delta$, which we set equal to $10^{-7}$. Note that this could be set equal to or different from the $\epsilon_{xi}$ of our algorithm (in general it seems most appropriate to set $\delta$ to be smaller than the $\epsilon_{xi}$).  

\begin{algorithm}[!ht]
\caption{Procedure to compute degree of approximation error of candidate ESSPM strategy}
\label{al:esspm-error} 
\textbf{Inputs}: \small Payoff matrix with $m$ pure strategies per player, candidate ESSPM $x^*$, degree of precision $\delta$ \normalsize
\begin{algorithmic}
\For {$i = 1$ to $m$} 
\If {$u_1(i,x^*) - u_1(x^*,x^*) > \delta$} 
\State $\theta[i] \gets u_1(i,x^*) - u_1(x^*,x^*)$ 
\ElsIf {$u_1(i,x^*) - u_1(x^*,x^*) > (-1) \cdot \delta$} 
\If {$u_1(i,i) - u_1(x^*,i) > 0$} 
\State $\theta[i] \gets u_1(i,i) - u_1(x^*,i)$ 
\EndIf 
\Else 
\State $\theta[i] \gets 0$ 
\EndIf 
\EndFor 
\Return $\max_i \theta[i]$ 
\end{algorithmic}
\end{algorithm}

\section{Applying the algorithm to solve the Mutation-Population Game}
\label{se:example}
We first implemented the algorithm using the CVXPY non-convex quadratically-constrained quadratic program (QCQP) Suggest-and-Improve framework~\cite{Park17:General}, since it is the only publicly-available software we are aware of for solving our QCMIP formulation (which is nonconvex). We considered all of their algorithms and found best performance with the coordinate-descent algorithm, because it focuses more on trying to find and maintain feasibility as opposed to the other algorithms which are more focused on optimality. This is appropriate for our problem since it is a feasibility program with no objective function. We create an initial solution for the algorithm by solving the semidefinite program (SDP) relaxation, which we then attempt to improve iteratively using coordinate descent. For the parameters we set $\epsilon_{xi} = 0.00001$ and $M_{xi} = 1 + \epsilon_{xi}$.

We tested this implementation on the well-studied Mutation-Population game (Figure~\ref{fi:mutation-population})~\cite{Maschler13:Game}. In this game there are two types of animals: hawks (who are aggressive), and doves (who are peaceful). When an animal invades the territory of another animal of the same species, a hawk will aggressively repel the invader, while a dove will yield and be driven out of the territory. The game has one symmetric Nash equilibrium in which Dove is selected with probability $\frac{1}{5}$ and Hawk $\frac{4}{5}$. This strategy profile also constitutes an ESS and an ESSPM. (Note that the game also has two asymmetric Nash equilibria (Dove, Hawk) and (Hawk, Dove), both of which are neither ESS nor ESSPM.)

\begin{figure}[!ht]
\begin{center}
\begin{tabular}{c|*{2}{c|}}
&Dove &Hawk \\ \hline
Dove &(4,4) &(2,8) \\ \hline
Hawk &(8,2) &(1,1) \\ \hline
\end{tabular}
\caption{The Mutation-Population Game}
\label{fi:mutation-population}
\end{center}
\end{figure}

When we ran our algorithm on this game, the initial solution obtained from the SDP relaxation was (0.637, 0.363), and the solution after the first iteration of the algorithm was $(0.245, 0.755)$. However, the solution varied drastically between iterations as the algorithm progressed and did not converge over 1000 iterations, though the degree of violation of the constraints generally decreased. We therefore concluded that the CVXPY-QCQP framework was not suitable for solving our problem, despite the fact that good performance has been observed for previous benchmark domains~\cite{Park17:General}).

We next decided to run experiments with Gurobi's mixed-integer programming solver~\cite{Gurobi14:Gurobi}. Gurobi does not have a solver for non-convex quadratically-constrained programs; however, we can apply a known technique that approximates products of variables using piecewise linearization by adding additional constraints known as ``SOS2'' constraints~\cite{Bisschop06:AIMMS}. (Note that Gurobi can solve programs that have quadratic objectives and constraints but only if the matrices are positive semi-definite, which is not the case for our problem.) This approach introduces an additional layer of approximation due to the piecewise-linear approximation of quadratic variables, and also introduces an additional parameter to use for the number of breakpoints in the linearization. However, it is able to quickly and effectively solve the mutation-population game. For example, using $k=20$ breakpoints with $\epsilon = \num{e-5}$ it outputs a solution of $(0.19972, 0.80028)$ in 0.242 seconds, which is extremely close to the optimal solution and produces an approximation error of $\num{1.4e-4}.$  

\section{Experiments}
\label{se:experiments}
Given the results of the preceding section we decided to use Gurobi's solver for the remainder of our experiments. First, in Section~\ref{se:experiments-mp} we further explore the effect of varying the choices of parameters $k$ and $\epsilon$ on the runtime and performance. In Section~\ref{se:exp-uniform} we look at performance on uniform-random symmetric 2x2 games. In addition to runtime and performance we also explore the false negative rate of our algorithm. It is known that all 2x2 games that satisfy a certain condition contain an ESS (and therefore ESSPM). The condition is that $a_{11} \neq a_{21}$ and $a_{12} \neq a_{22}$, where $A$ is the payoff matrix~\cite{vanDamme87:Stability}. This implies that an ESSPM will exist with probability 1 in uniformly-generated 2x2 games, since there is probability zero on the exact payoff values needed being equal. Therefore, any infeasibility output of our algorithm will be a false negative. This does not hold for games with more than two strategies, as the rock-paper-scissors example in the introduction demonstrates. At the other extreme, it has been shown that as the number of pure strategies approaches infinity the probability of existence of an ESS with support of size two converges to 1 for games with payoffs chosen according to a distribution with ``exponential and faster decreasing tails'' which includes uniform (the \emph{support} of a randomized strategy is the set of pure strategies that are played with nonzero probability)~\cite{Hart08:Evolutionarily}. In Section~\ref{se:exp-uniform} we also explore 3x3 games and discuss further scalability; in Section~\ref{se:ex-gamut} we explore a generalized game class generated by the GAMUT repository~\cite{Nudelman04:Run}; and in Section~\ref{se:ex-cancer} we conduct experiments on a recently-proposed 4x4 game model of cancer~\cite{Hurlbut18:Game}.  

\subsection{Mutation-Population Game}
\label{se:experiments-mp}
In the Mutation-Population Game we consider the effect of selection of parameters $k$ (number of breakpoints) and $\epsilon$ on the runtime and error (computed by Algorithm~\ref{al:esspm-error}). Table~\ref{ta:mp-fixed-epsilon} shows the effect of varying $k$ keeping $\epsilon = 10^{-5}$ fixed. For the values of $k$ up to 30 the runtime is under one second, before increasing to 1.7 seconds for $k = 50$ and 6.6 seconds for $k = 100$. Interestingly for $k = 5$ our algorithm outputs infeasibility, while it produces low errors for larger $k$ which decrease monotonically. This demonstrates that choosing too few breakpoints can be problematic and lead to a false negative. Using 10--30 breakpoints leads to low runtimes and relatively small errors. These experiments were using up to 64 cores (threads) for the algorithm (though the algorithm also performs competitively with significantly fewer threads).

\begin{table*}[!ht]
\centering
\begin{tabular}{|*{7}{c|}} \hline
\# Breakpoints &5 &10 &20 &30 &50 &100\\ \hline
Time(s) &0.04 &0.08 &0.13 &0.54 &1.69 &6.6\\ \hline
Error &INFEASIBLE &0.001 &\num{1.4e-4} &\num{5.5e-5} &\num{1.3e-5} &\num{5.9e-6}\\ \hline
\end{tabular}
\caption{Running time 
and approximation error for Mutation-Population Game for different breakpoints with $\epsilon = 10^{-5}$.}
\label{ta:mp-fixed-epsilon}
\end{table*}
\normalsize

We next explored the effect of varying $\epsilon$ while keeping $k = 20$ fixed (Table~\ref{ta:mp-fixed-breakpoints}). Changing $\epsilon$ seems to have minimal effect on the runtime for this game. For the three largest $\epsilon$ values our algorithm outputs that the model is infeasible. This demonstrates that careful selection of $\epsilon$ in addition to $k$ is needed to prevent a false negative. The results also show that the error does not necessarily decrease monotonically with smaller $\epsilon$ (though this conclusion is drawn from a small sample size).   

\begin{table*}[!ht]
\centering
\begin{tabular}{|*{7}{c|}} \hline
$\epsilon$ &\num{1e-1} &\num{1e-2} &\num{1e-3} &\num{1e-4} &\num{1e-5} &\num{1e-6}\\ \hline
Time(s) &0.13 &0.31 &0.28 &0.35 &0.13 &0.25\\ \hline
Error &INFEASIBLE &INFEASIBLE &INFEASIBLE &\num{6.6e-5} &\num{1.4e-4} &\num{1.9e-4}\\ \hline
\end{tabular}
\caption{Running time 
and approximation error for Mutation-Population Game for different values of $\epsilon$ using 20 breakpoints.}
\label{ta:mp-fixed-breakpoints}
\end{table*}
\normalsize

\subsection{Uniform random games}
\label{se:exp-uniform}
We next considered 2x2 symmetric games with all payoffs generated uniformly at random (we generated 4 payoffs for the matrix $A$ for player 1, and then assume that the payoff matrix for player 2 is $A^T$). We generated and solved 10,000 games, using 20 breakpoints with $\epsilon = 10^{-5}$. Of these, the pre-processing step confirmed that 7506 had a pure-strategy ESSPM. Of the remaining games, our algorithm computed an optimal solution for 2454, and output that 40 were infeasible. The average running time for the games with an optimal (non-pure) solution was 0.18s, and the average running times for the infeasible cases was 0.26s. The average approximation error for the optimal solutions was $\num{1.4e-4}.$ 

These results lead to several immediate observations. First, a very large percentage games contain a pure-strategy ESSPM (close to 75\% in our sample). Note that this is not surprising, since we are considering symmetric games. If the payoff matrix is $A$, then the game has a pure-strategy ESSPM whenever $a_{11} > a_{21}$ or $a_{22} > a_{12}$, which is exactly 75\%. Note that for randomly-generated games the possibility of identical payoffs for different strategy profiles has probability zero. This may make it difficult for the second condition of Theorem~\ref{th:esspm} to hold (the probability of it holding when $x^*$ is a pure strategy is 0, though it could be possible for mixed strategies). However, rather than trying to optimize our algorithm just for random games we elect to keep our general-purpose framework as the second condition can be important for many games.

The observation that our algorithm outputs infeasibility for 40 games (which constitutes 1.6\% of the games in our sample that did not contain a pure-strategy ESSPM) shows that false negatives can occur even for reasonably chosen values of the parameters $k = 20$ and $\epsilon = \num{e-5}.$ On the positive side, the runtimes and average approximation errors are quite low. So for the games that our algorithm does output a valid solution, the solution is obtained quickly with low error. And furthermore a valid solution is found for a very large percentage of the game instances (98.4\% of the games that do not contain a pure-strategy ESSPM). If our algorithm outputs infeasibility for a specific game we are interested in, we can always solve it again with larger values of $k$ and/or smaller values of $\epsilon$ until a solution is found (note that this would only work for $2 \times 2$, and for larger games we could not be sure whether a solution exists and may have to give up eventually).

We also experimented on 3x3 uniform-random symmetric games using the same parameter values $k = 20$, $\epsilon = 10^{-5}$. Of the 1000 games in our sample, our algorithm output that 709 had a pure-strategy ESSPM, it found an optimal solution for 39, and output that 252 were infeasible.  The average running time for the games with an optimal (non-pure) solution was 12.7s, and the average running times for the infeasible cases was 14.7s. The average approximation error for the optimal solutions was $\num{3.3e-4}.$ 

These results indicate as before that a large fraction of random 3x3 games have a pure-strategy ESSPM (the theoretical value should be $\frac{19}{27} \equiv 0.704$, and in general for $m$ pure strategies should be $1 - \left( \frac{m-1}{m} \right)^m$). In contrast to the $m = 2$ results, for $m = 3$ we observe that our algorithm outputs infeasibility for a very large fraction of the games with no pure-strategy ESSPM (86.7\% in contrast to 1.6\% for $m = 2$ using the same parameter values). Unlike the $m=2$ case we can not be sure for a given game whether our algorithm has produced a false negative or whether no ESSPM exists, so we cannot conclude whether this disparity is due to an increased false negative rate or nonexistence of ESSPM. For the games that our algorithm did solve optimally the running time was only 12.7s on average (and the running time to determine infeasibility was only 14.7s on average), and the approximation error was quite low for the games which our algorithm output an optimal solution.

\renewcommand{\tabcolsep}{1pt}
\begin{table*}[!ht]
\centering
\footnotesize
\begin{tabular}{|*{11}{c|}} \hline
\#Strategies &\#Games &k &$\epsilon$ &\#Pure &\#Optimal &\#Infeasible &Optimal avg. runtime &Infeasible avg. runtime &Avg. error\\ \hline
2 &10,000 &20 &$10^{-5}$ &7506 &2454 &40 &0.18s &0.26s &\num{1.4e-4}\\   \hline
3 &1000 &20 &$10^{-5}$ &709 &39 &252 &12.7s &14.7s &\num{3.3e-4}\\ \hline
3 &1000 &10 &$10^{-5}$ &712 &68 &220 &1.2s &1.7s &0.002\\ \hline
4 &1000 &10 &$10^{-5}$ &671 &48 &281 &6.3s &9.0s &0.004\\ \hline
5 &1000 &10 &$10^{-5}$ &656 &50 &293 &29.7s &70.2s &0.005\\ \hline
\end{tabular}
\caption{Results for uniform random games.}
\label{ta:random}
\end{table*}
\normalsize

\renewcommand{\tabcolsep}{4pt}

We also performed experiments for 3x3 games using $k=10$, with $\epsilon = 10^{-5}$ as before. Over 1000 games our algorithm output that 712 had a pure-strategy ESSPM, found an optimal solution for 68, and output that 220 were infeasible. The average running time for the games with an optimal (non-pure) solution was 1.2s, and the average running times for the infeasible cases was 1.7s. The average approximation error for the optimal solutions was $0.002.$ Using fewer breakpoints decreased the infeasibility rate over this sample from 86.7\% to 76.4\%, while leading to a significant reduction in running time, but an order of magnitude increase in approximation error.  

We also experimented on 4x4 games using $k=10$, $\epsilon = 10^{-5}$. Over 1000 games the algorithm output that 671 had a pure-strategy ESSPM, found an optimal solution for 48, and output that 281 were infeasible. The average running time for the games with an optimal (non-pure) solution was 6.3s, and the average running times for the infeasible cases was 9.0s. The average approximation error for the optimal solutions was $0.004.$ For 5x5 games with the same parameters we found 656 pure-strategy ESSPM, 50 optimal solutions, and 293 infeasible. The average runtime for the optimal cases was 29.7s, and 70.2s for the infeasible cases. The average error for the optimal solutions was $0.005$.   


\subsection{Experiments on games from GAMUT repository}
\label{se:ex-gamut}
In addition to uniform randomly-generated games we also experimented on the class of Chicken games created using the GAMUT generator~\cite{Nudelman04:Run}. These games generalize the Mutation-Population game that was previously considered. For a payoff matrix $A$, Chicken games satisfy $a_{21} > a_{11} > a_{12} > a_{22}$. We generated 1,000 games from this class where, as before, we normalize payoffs to be between 0 and 1 before applying our algorithm. We only consider Chicken games with $m = 2$ pure strategies per player, as GAMUT does not support a generalized version. Results using $k = 20, \epsilon=\num{e-5}$ are given in Table~\ref{ta:gamut}. Note that none of the games have a pure-strategy ESSPM, while our algorithm incorrectly outputs infeasibility for 1.8\% of the games (we know that an ESSPM exists in nondegenerate 2x2 games). The average runtimes and errors of our algorithm are low for this class.

\begin{table*}[!ht]
\centering
\begin{tabular}{|*{7}{c|}} \hline
\#Pure &\#Optimal &\#Infeasible &Optimal avg. runtime &Infeasible avg. runtime &Avg. error\\ \hline
0 &982 &18 &0.13s &0.19s &\num{1.3e-4}\\ \hline
\end{tabular}
\caption{Results for 1,000 games of Chicken generated by GAMUT, using $k = 20$, $\epsilon = 10^{-5}$.} 
\label{ta:gamut}
\end{table*}
\normalsize

Several popular games such as Prisoner's Dilemma and Battle of the Sexes contain a pure-strategy ESS (and therefore ESSPM), so experimenting on them is not useful. GAMUT also supports a class called ``Hawk and Dove,'' however the way this is generated produces only games that have a pure-strategy ESS. (Note that several different definitions of ``Hawk and Dove games'' have been proposed in the literature. For example, while the Population-Mutation game we considered does not fall into GAMUT's class, it is considered a Hawk-Dove game by Maschler et al.~\cite{Maschler13:Game}.)

\subsection{Experiments on cancer game models}
\label{se:ex-cancer}
Our final set of experiments is on a game recently proposed as a model for cancer~\cite{Hurlbut18:Game}. Each player represents a tissue cell and can adopt one of the four phenotypes that might be present within a cancer tumor.

\begin{quotation}
Stromal cancer cells (A-) have no particular benefit or cost unique to themselves, and they are considered a baseline neutral cell within the context of the model. In contrast, angiogenesis-factor producing cells (A+) vascularize the local tumor area which consequently introduces a nutrient rich blood to the benefit of all interacting cells. Nutrient recruitment expands when A+ cells interact with one another. Cytotoxic cells (C) release a chemical compound which harms heterospecific cells and increases their rate of cell death. The cytotoxic cells benefit from the resulting disruption in competition caused by the interaction. For simplicity, our model presumes that cytotoxic cells are themselves immune to this class of agent. Finally, proliferative cells (P) possess a reproductive or metabolic advantage relative to the other cell types. In our model this advantage does not compound with the nutrient enrichment produced by vascularization when A+ cells are present; however, it does place the proliferative cell at a greater vulnerability to cytotoxins~\cite{Hurlbut18:Game}.
\end{quotation}

The payoffs are given by Table~\ref{ta:cancer}, where the parameters are defined in Table~\ref{ta:parameters} (recall that the algorithm will normalize the payoffs to fall in [0,1]).

\begin{table}[!ht]
\begin{center}
\begin{tabular}{c|*{4}{c|}}
&A- &A+ &P &C \\ \hline
A- &1 &1+d &1 &1-c \\ \hline
A+ &1-a+d &1-a+d+f &1-a+d &1-c-a+d \\ \hline
P &1+g &1+d+g &1+g &(1+g)(1-c) \\ \hline
C &1-b+c &1-b+d+e &1-b+e &1-b \\ \hline
\end{tabular}
\caption{Payoff matrix of symmetric game with four cancer cell phenotype strategies.}
\label{ta:cancer}
\end{center}
\end{table}

\begin{table}[!ht]
\begin{center}
\begin{tabular}{cl}
Parameter &Interpretation \\ \hline
a & Cost of producing angiogenesis factors \\ 
b & Cost of producing cytotoxin \\ 
c & Cost of interaction with cytotoxin \\
d & Resource benefit when interacting with A+ \\ 
e & Exploitation benefit for C when cytotoxin damages others \\ 
f & Synergistic resource benefit when two A+ cells interact \\ 
g & Reproductive advantage of P cell \\ \hline 
\end{tabular}
\caption{Parameters used within the game model.}
\label{ta:parameters}
\end{center}
\end{table}

That paper does not specify values for the parameters, so we select them all to be uniformly at random within the interval [0,0.5], which ensures that all payoffs are nonnegative. Our results over 1,000 randomly generated games within this class are given in Table~\ref{ta:cancer-results}.

\begin{table*}[!ht]
\centering
\small
\begin{tabular}{|*{9}{c|}} \hline
$k$ &$\epsilon$ &\#Pure &\#Optimal &\#Infeasible &Optimal avg. runtime &Infeasible avg. runtime &Avg. error\\ \hline
10 &$10^{-5}$ &869 &58 &73 &4.08s &9.36s &0.007\\ \hline
20 &$10^{-5}$ &867 &11 &122 &82.58s &152.61s &0.002\\ \hline 
\end{tabular}
\caption{Results for 1,000 cancer games generated with game parameters uniformly in [0,0.5].}
\label{ta:cancer-results}
\end{table*}
\normalsize

\section{Conclusion}
\label{se:conclusion}
We have presented an algorithm for computation of a restriction of Evolutionarily Stable Strategies (ESS) where robustness is guaranteed only against pure-strategy mutations and not necessarily against all mixed-strategy mutations, which we call Evolutionarily Stable Strategies against Pure Mutations (ESSPM). ESS is a well-studied refinement of Nash equilibrium that is biologically-motivated and has been applied to cancer. Previously known computational results for ESS computation are negative hardness results and algorithms based on brute-force support enumeration. Our algorithm is based on a novel quadratically-constrained mixed integer feasibility program formulation. 


We experimented with different parameter settings for the Population-Mutation Game, which were then used for experiments in uniform-random games, in a generalized class of Chicken games created by the GAMUT generator, and in a recently proposed game model of cancer. We observe that for certain classes a large fraction of games has a pure-strategy ESSPM, and that our algorithm can sometimes output a false negative infeasibility for 2x2 games where an ESS is known to exist. However, in several experiments the false negative rate is small, and the running times and errors produced by our algorithm are generally quite low for the games that are solved. We considered two commercial solvers CVXPY-QCQP and Gurobi, focusing our experiments on Gurobi after preliminary analysis indicated that it was better suited for our setting. The algorithm takes two input parameters, the number $k$ of breakpoints for the piecewise-linear approximations of variable products, and the degree $\epsilon$ of constraint approximation to use. Our experiments investigated the effect of selection of these parameters on performance.  

We would like to further tune the parameters for epsilon and the number of breakpoints and further explore the cases of infeasibility that we encountered. It is not clear for the games with more than 2 actions whether the infeasibility was due to the approximation error of the algorithm or to the lack of existence of an ESSPM. Eventually we would like to generalize our approach to compute an ESS that is robust to all mutation strategies (as in the standard definition of ESS), not just to pure strategy mutations as we have done. It is possible that this can be accomplished within the same mixed-integer non-convex quadratically-constrained optimization framework. Our approach could potentially apply to generalizations to multiple players by approximating the products of variables in a similar way. We would also like to further explore scalability and to experiment on problems with real-world applicability. We suspect that performance can be significantly improved using an improved solver for our optimization formulation (this could potentially be accomplished without the need for the piecewise-linear approximation). We are particularly interested in considering more complex cancer models. We expect for larger real-world problems that have ESS with medium-sized support that our algorithm would outperform the existing algorithms such as Haigh's which iterates over all possible supports, as has been demonstrated to be the case for Nash equilibrium algorithms. 


\bibliographystyle{plain}
\bibliography{D://FromBackup/Research/refs/dairefs}

\end{document}